\documentclass[preprint,12pt]{elsarticle}
\usepackage{amsmath}
\usepackage{amssymb}
\usepackage{graphicx}
\usepackage{subfig} 
\usepackage{amsthm} 
\usepackage{bbm} 
\usepackage{enumerate}
\usepackage{natbib}
\usepackage[ruled]{algorithm2e}
\usepackage{url} 
\usepackage{float}
\usepackage{titlesec}
\usepackage{xcolor}

\usepackage{comment}

\newtheorem{theorem}{Theorem}

\newtheorem{Lemma}{Lemma}
\newtheorem{Corollary}{Corollary}

\journal{Computational Statistics \& Data Analysis}

\begin{document}

\begin{frontmatter}

\title{Exact Designs for Order-of-Addition
Experiments Under a Transition-Effect
Model}

\author{Jiayi Zheng, Nicholas Rios} 

\affiliation{organization={George Mason University},
            addressline={4400 University Drive}, 
            city={Fairfax},
            postcode={22030}, 
            state={VA},
            country={United States}}

\bigskip
\begin{abstract}
In the chemical, pharmaceutical, and food industries, sometimes the order of adding a set of components has an impact on the final product. These are instances of the order-of-addition (OofA) problem, which aims to find the optimal sequence of the components. Extensive research on this topic has been conducted, but almost all designs are found by optimizing the $D-$optimality criterion. However, when prediction of the response is important, there is still a need for $I-$optimal designs. A new model for OofA experiments is presented that uses transition effects to model the effect of order on the response, and the model is extended to cover cases where block-wise constraints are placed on the order of addition. Several algorithms are used to find both $D-$ and $I-$efficient designs under this new model for many run sizes and for large numbers of components. Finally, two examples are shown to illustrate the effectiveness of the proposed designs and model at identifying the optimal order of addition, even under block-wise constraints. 

\end{abstract}

\begin{keyword}
Design of Experiments \sep Transition-Effect Model \sep I-optimality \sep Simulated Annealing

\end{keyword}

\end{frontmatter}

\newpage
\section{Introduction}
\label{sec:intro}
In the chemical industry, the process to output a product typically requires several components to be added. The order of adding the components often influences the quality of the output, and searching for the optimal order is of concern. This problem is known as the Order of Addition (OofA) problem. For example, as suggested in \citet{chandrasekaran2006}, the order of different alcohols affects the  synthesis of disubstituted carbonates. Similarly, in the field of engineering, \citet{joiner1976} found that the calibration order of different cryogenic flows influences the precision of the cryogenic meter due to its sensitivity to temperature changes and the cost of the overall experiment. In social sciences, \citet{miller1998} found that candidate name order could have an impact on election results. In job scheduling problems, the order of performing jobs directly affects the total job cost \citet{zhao2021}. In \citet{YangSun2020}, the order of four drugs was varied to study the impact of drug order on the treatment of lymphoma.  Overall, the OofA problem is becoming a multi-disciplinary topic requiring further research.

In the OofA problem, when the number of components is small, it is easy to test all possible permutations. However, the experiments become extremely expensive even with a slight increase in the number of components. If there are $m$ components in total, then the number of all possible permutations is given by $m!$. For example, when there are six components, there are $6!=720$ permutations in total. In this article, we denote a single permutation of $m$ components by $\textbf{a} = (a_1, \dots, a_m)$, and we denote the set of all $m!$ permutations by $D_m$.

How should an experiment be designed when it is impossible to test all the permutations? This article focuses on using the $D-$ and $I-$optimality criteria to choose an experimental design. While $D-$optimal OofA designs have been studied extensively, the $I-$optimality criterion has not received much attention in the OofA literature. An experimental design is a subset of the $m!$ orderings to test. The $I-$optimality criterion minimizes the integrated prediction variance over all feasible permutations of the $m$ components. Much of the existing literature on OofA experiments focuses on the $D-$optimality criterion, which produces designs that aid with parameter estimation. For example, \citet{lin2019order} showed a systematic method for constructing fractional $D-$optimal OofA designs for the pairwise order model, and \citet{chen2020construction} proposed a modified construction method for $D-$optimal designs based on block designs. \citet{mee2020order} examined the $D-$efficiency of a triplets model that allowed for interactions between pairwise order terms. 

The focus on $I-$optimality is of particular interest when there are constraints on the order of addition that limit the number of possible orders to examine. In real applications, there are often constraints on the order of some components. For example, in Solvay process of chemical industry, which is used to produce sodium carbonate ($\text{Na}_2\text{CO}_3$), the addition of $\text{NH}_3$-concentrated solution must be added before $\text{CO}_2$. According to \citet{WangLi2019}, the industry requires the concentrated solution to be ammoniated with the brine (NaCl) solution first. In such cases, a constrained order-of-addition problem arises, with pairwise constraints placed on the orders of addition. Another example is given by scheduling problems, where certain tasks may need to be completed before others can begin. In this paper, we will examine cases where the components are grouped into blocks, and the blocks must be executed in a known order. 

The rest of this paper is organized as follows. Section \ref{sec:prelims} will review existing methods or the OofA problem. Then, Section \ref{sec:methods} will propose a new OofA model based on transitions between components in a permutation, and also explain the algorithms we use in the article to find exact designs with high $I-$efficiency and $D-$efficiency. Section \ref{sec:simuresults} will compare the performance of these algorithms at finding highly efficient designs under these criteria when the number of components $m$ is large, and traditional exchange algorithms are inefficient. Finally, Section \ref{sec:example} will provide an example to illustrate the implementation of the proposed methods.

\section{Preliminaries}
\label{sec:prelims}
In order to discuss the proposed methodology, it is helpful to first review previous research results.
\subsection{PWO model by Van Nostrand}
Multiple models have been proposed for the OofA problem. For example, the component-position model focuses on the absolute positions of the components, which was suggested by \citet{YangSun2020}. The pairwise order (PWO) model was first suggested in \citet{van1995}. Consider an experiment when there are $m$ components, where $m$ is any positive natural number. Label the components with numbers from one to $m$. Then for any pair of components $i,\ j$, where $1\leq i< j\leq m$, define the PWO factor $x_{ij}$ as 
\begin{align}
    x_{ij}(\textbf{a})=
\begin{cases}
    1\ \ \ \text{if component $i$ is added before $j$ in $\textbf{a}$}\\
    -1\ \ \ \text{if component $i$ is added after $j$ in $\textbf{a}$}\\
\end{cases}
\end{align}

For example, suppose $m=3$. In this case, there would be six possible permutations in total, i.e. $D_3 = \{(1,2,3),\ (1,3,2),\ (2,1,3),\ (2,3,1),\ (3,1,2),\ (3,2,1)\}$. $D_3$ can be written as a $6\times3$ matrix, which can be seen in the left side of Table \ref{tab1}.\\
When $\textbf{a}=(1,3,2)$, $x_{12} = 1$ and $x_{23} = -1$. Meanwhile, when $\textbf{a}=(3,1,2)$, $x_{13} = x_{23} = -1$. The PWO model can be written as 
\begin{align}
    E(y|\textbf{a})=\beta_0+\sum_{i=1}^{m-1}\sum_{j=i+1}^m\beta_{ij}x_{ij}(\textbf{a})
    \label{eqn:PWOmodel}
\end{align}
This model is denoted as the \textit{main effects PWO model} in \citet{voelkel2019}, and as the \textit{simple pairwise model} in \citet{mee2020order}. Optimal designs for the PWO model were constructed by \citet{OOADesign}, who also considered a tapered PWO model. There are other versions of the model, such as the \textit{triplets order-of-addition model} discussed in \citet{mee2020order}, and the \textit{special cubic Scheffe model} in \citet{becerra2021bayesian} which both consider the triple interaction of components' orders. In this article, only the \textit{simple pairwise model} is discussed, and it is denoted as the PWO model. There are $\binom{m}{2}$ PWO factors in total in Model (\ref{eqn:PWOmodel}). Model (\ref{eqn:PWOmodel}) can be rewritten as $ E(y|\textbf{a})=\Tilde{x}^T\Tilde{\beta} $, where as
$$\Tilde{x}=\begin{pmatrix}
1\\
x_{12}(\textbf{a}) \\
x_{13}(\textbf{a}) \\
\vdots\\
x_{(m-1)m}(\textbf{a})
\end{pmatrix}\ \ \Tilde{\beta}=\begin{pmatrix}
\beta_0\\
\beta_{12} \\
\beta_{13} \\
\vdots \\
\beta_{(m-1)m}
\end{pmatrix}$$
From here it can be seen that the PWO model is a special case of multiple linear regression model. The smallest
number of observations required to fit the model is $\binom{m}{2}+1$. A matrix $X$ can be created based on the vectors of $\Tilde{x}$, i.e. $X=(\Tilde{x}_1^T,\Tilde{x}_2^T,\dots,\Tilde{x}_{m!}^T)$. This $X$ matrix is used to calculate different optimality criteria and therefore determine the optimal design. 
\begin{table}[!h]
\begin{center}
\begin{tabular}{||c c c c c c||} 
 \hline
 \multicolumn{3}{||c}{OofA $D_3$} & \multicolumn{3}{c||}{PWO Model $X_3$}\\ [0.5ex] 
 \hline\hline
 & & & $x_{12}$ & $x_{13}$ & $x_{23}$\\
  1 & 2 & 3\ \ \ \ \ \  & 1 & 1 & 1\\ 
  1 & 3 & 2\ \ \ \ \ \   & 1 & 1 & -1\\
  2 & 1 & 3\ \ \ \ \ \   & -1 & 1 & 1\\
  2 & 3 & 1\ \ \ \ \ \   & -1 & -1 & 1\\
  3 & 1 & 2\ \ \ \ \ \   & 1 & -1 & -1\\
  3 & 2 & 1\ \ \ \ \ \   & -1 & -1 & -1\\[1ex]
 \hline
\end{tabular}
\caption{OofA design matrix and PWO model when $m=3$}
\label{tab1}
\end{center}
\end{table}

\subsection{Optimality Criteria}
\label{subsec:preopt}
In previous research, different criteria have been used to determine the best designs for OofA experiments. The most popular is $D-$optimality, which maximizes the determinant of the moment matrix; this was discussed in \citet{becerra2021bayesian}, \citet{sambo2014}, \citet{li2021efficient}. Similarly, other optimality criteria have been examined, such as
$Q-$optimality, $A-$optimality in \citet{li2021efficient}, $G-$optimality, and $EI-$optimality in \citet{li2021efficient}. Let $X_f$ be the matrix of PWO variables for every feasible permutation. Let $X_n$ be a matrix of $n$ rows from $X_f$. A commonality among the criteria is that they are based on the moment matrix $M$, defined as $M=X_n^TX_n/n$, where $n$ is the number of runs in a PWO design. For parameter estimation, we focus on $D-$optimality, i.e., $|M|^{1/q}$, where $q$ is the number of columns of $M$, where larger values of the $D-$optimality criterion are preferred. For prediction variance, we choose to focus in minimizing the $I-$optimality criterion, which minimizes the average prediction variance over an experimental region. \citet{goos2016optimal} noted that $G-$optimal designs seek to minimize the maximum prediction variance over the experimental region. Since $G-$optimal designs only minimize the maximum prediction variance, this article focuses on $I-$optimality instead, also known as $V-$optimality in \citet{goos2014}. \citet{becerra2021bayesian} suggested that the $I-$optimality criterion has advantages compared to other criterion, as it focuses on obtaining precise results. This criterion aims at minimizing the average variance in our predicted response over all feasible permutations under a particular set of pairwise constraints. Let the $I-$optimality criterion be $I= \text{trace}(M^{-1}X_f^TX_f)$. Smaller values of $I$ indicate better designs. 

In the OofA literature, $I-$optimal designs have been previously examined for main effect models, but not much attention has been given to designs for models with higher-order interactions. Notably, \citet{schoen2023order} showed that any design that is an orthogonal array of strength 2 in $m$ components is an $I-$optimal design for the main effects PWO model. They did examine the relative $D-$efficiency of certain orthogonal arrays in a model with interaction terms, but the $I-$efficiency has yet to be investigated. Also, these orthogonal arrays do not exist in all sizes; for example, \citet{schoen2023order} acknowledge that OofA orthogonal arrays of strength 2 only exist for run sizes that are a multiple of 12. It is important to develop flexible methods that allow arbitrary run sizes to accommodate for constrained experimental budgets.

Denote the $I-$optimality of the full constrained design matrix as $I_{full}$, and the $D-$optimality of the full design matrix as $D_{full}$. For the design matrix corresponding to a subset of the full design with an $I-$optimality of $I_{reduced}$, one can find the relative $I-$efficiency as $\frac{I_{full}}{I_{reduced}}$. The relative $I-$efficiency is written this way so that higher relative efficiencies are better for the reduced design. For $D-$optimality (since larger values are better), the relative $D-$efficiency is $\frac{D_{reduced}}{D_{full}}$.

\section{Proposed Methods}
\label{sec:methods} 

\subsection{A Transition Effect (TE) Model}

Suppose there are $m$ components to be sequentially added. Let $\textbf{a} = (a_1, a_2, \dots, a_m)$ be a permutation of $(1,\dots,m)$. Let $\tau(\textbf{a})$ be the expected response the order $\textbf{a}$. Model (\ref{eqn:oofamodel}) is proposed for the OofA experiment, which is given by \begin{align}
\label{eqn:oofamodel}
    \tau( \textbf{a}) =  \beta_0 + \sum_{j = 1}^m \sum_{k = 1, k \neq j}^m \beta_{j,k} x_{j,k}(\textbf{a}),
\end{align} where $x_{(j,k)}(\textbf{a}) = 1$ if component $k$ directly follows component $j$ in $\textbf{a}$, and 0 otherwise. In model (\ref{eqn:oofamodel}), $\beta_{j,k}$ represents the effect of adding component $k$ directly after adding component $j$ on the expected response. We refer to this as the ``transition effect'' from component $j$ to $k$. As a small example, suppose $m = 3$ and $\textbf{a} = (3,1,2)$. Then $\tau(\textbf{a}) = \beta_0 + \beta_{3,1} + \beta_{1,2}$. These coefficients can be estimated using least squares. Constraints are required to estimate the parameters in Model (\ref{eqn:oofamodel}). This is because $\sum_{j=1}^m\sum_{k=1, k \neq j}^m x_{j,k}(\textbf{a}) = m-1$ for every $\textbf{a} \in D_m$, since $m-1$ transitions must occur in a permutation of $m$ components, which creates a linear dependency with the intercept. For this reason, we assume that $\delta_{m,m-1} = 0$. This model has $2{m \choose 2}$ estimable parameters.

Model (\ref{eqn:oofamodel}) assumes that as each component is added, the effect of the permutation on the response is a sum of the effects of direct transitions; i.e., transitions between adjacent components in the permutations. A generalization of this model would be to include effects that represent transitions between components that are further apart. Suppose, for a permutation $\textbf{a} = (a_1, \dots, a_s, \dots, a_u, \dots, a_m)$, that $a_s = j$ and $a_u = k$. Define the transition length from $j$ to $k$ as $d(j,k,\textbf{a}) = u - s$. This transition length is not symmetric; for example, if $\textbf{a} = (3,1,2,4)$, then $d(3,2, \textbf{a}) = 2$, but $d(2,3, \textbf{a}) = -2$. A model that includes transitions of lengths 1 and 2 is \begin{align}
\label{eqn:oofamodel2}
    \tau(\textbf{a}) = \beta_0 + \sum_{j=1}^m\sum_{k=1, k \neq j}^m \Big[ \beta_{j,k}x_{j,k}(\textbf{a}) +  \delta_{j,k}x_{j,k}^{(2)}(\textbf{a})\Big],
\end{align} where $x_{j,k}^{(2)}(\textbf{a}) = 1$ if $d(j,k,\textbf{a}) = 2$, and $x_{j,k}^{(2)}(\textbf{a}) = 0$ otherwise. In this case, $\delta_{j,k}$ represents the effect of the length two transition from component $j$ to $k$. Model (\ref{eqn:oofamodel2}) has $2m(m-1)$ parameters. However, since every permutation of $(1,\dots,m)$ must have $m-2$ transitions of length 2, it follows that $\sum_{j=1}^m\sum_{k = 1, k \neq j}^m x_{j,k}^{(2)}(\textbf{a}) = m - 2$, which would lead to scenarios where the model matrix is not full rank. Therefore, we set $\delta_{m,m-1} = \beta_{m,m-1} = 0$ to ensure identifiability for $m \geq 5$ for this model. 

A useful initial result for constructing optimal designs under the transition effect model is given below in Corollary \ref{cor:fullopt}.

\begin{Corollary}[Full Design is Optimal]
\label{cor:fullopt}
Let $D_m$ be the full design that uses each of the $m!$ orders exactly once. Let $X_m$ be the model matrix expansion of $D_m$ under model (\ref{eqn:oofamodel}) or model (\ref{eqn:oofamodel2}), and let $M_f = (1/m!)X_m^TX_m$. Let $\phi$ be an optimality criterion that is concave and permutation invariant. Then, $\phi(M_f) \geq \phi(M(w))$ for any design measure $w$ over $D_m$.
\end{Corollary}

Corollary \ref{cor:fullopt} shows that the full design, which uses each of the $m!$ possible orders exactly once, is $\phi-$optimal for any criterion $\phi$ that is concave and permutation invariant under the length-one and length-two transition effect models. The method of proving this result is very similar to Theorem 1 of \citet{OOADesign}. To evaluate the quality of a design, it is important to be able to quickly find its $D-$ or $I-$efficiency. Finding the $D-$ or $I-$efficiency of the full design could be done by enumerating all $m!$ permutations of the $m$ components, finding the corresponding moment matrix, and then evaluating the criterion, which involves a costly matrix operation (determinant for $D$, or inverse for $I$). To save time, a closed form of the moment matrix of the full design is derived in Theorem \ref{thm:momentmat1}.

\begin{theorem}[Moment Matrix for the Full Design Under Model (\ref{eqn:oofamodel})]
\label{thm:momentmat1}
Let $D_m$ be the full design that uses all $m!$ permutations only once. Let $X_m$ be the model matrix expansion of $D_m$ under Model (\ref{eqn:oofamodel}), with the constraint that $\beta_{m,m-1} = 0$. Let $q = 2{m \choose 2} - 1$. Then it follows that \begin{align}
    X_m^TX_m/m! = \begin{bmatrix}
        1  & (1/m)\textbf{1}^T  \\
        (1/m)\textbf{1}  & [(m-1)!I_q + (m-2)!V]/m!
    \end{bmatrix},
\end{align} where $\textbf{1}$ is a conformable column vector of 1s, and $V$ is a $q \times q$ matrix whose columns are indexed by the pairs $(1,2), \dots, (1,m), (2,1), \dots, (2,m), \dots, (m,1), \dots, (m,m-2)$, with elements \begin{align*}
    V_{(i,j), (k,\ell)} = \begin{cases}
        0 & \,\, \text{if } i = k \text{ or } j = \ell \text{ or } i = \ell, j = k, \\
        1 & \,\, \text{otherwise}.
    \end{cases}
\end{align*}
\end{theorem} 

Theorem \ref{thm:momentmat1} is helpful for finding the $D-$ and $I-$efficiencies of the full design without enumerating all $m!$ possible permutations. A similar result is shown for Model (\ref{eqn:oofamodel2}) using Theorem \ref{thm:momentmat2} below.

\begin{theorem}[Moment Matrix for the Full Design Under Model (\ref{eqn:oofamodel2})]
\label{thm:momentmat2}
Let $D_m$ be the full design that uses all $m!$ permutations only once. Let $X_m$ be the model matrix expansion of $D_m$ under Model (\ref{eqn:oofamodel2}), with the identifiability constraints $\beta_{m,m-1} = \delta_{m,m-1} = 0$. Then it follows that \begin{align}
    X_m^TX_m/m! = \begin{bmatrix}
        1  & (1/m)\textbf{1}^T & \frac{m-2}{m(m-1)}\textbf{1}^T  \\
        (1/m)\textbf{1}  & [(m-1)!I_q + (m-2)!V]/m! & [(m-3)!/m!]Q \\
        \frac{m-2}{m(m-1)}\textbf{1} & [(m-3)!/m!]Q^T   & [(m-2)(m-2)!I_q + R]/m!
    \end{bmatrix},
\end{align} where $V$ is as defined in Theorem \ref{thm:momentmat1}, the elements of $Q$ are \begin{align}
    Q_{(i,j),(k,\ell)} = \begin{cases}
        m-2  & \text{if } i = k, j \neq \ell \text{ or }  i \neq k, j = \ell, \\
        m-3  & \text{if } i \neq \ell, j = k,  \text{ or } i = \ell, j \neq k,\\
        m-4  & \text{if } \{i,j\} \cap \{k,\ell\} = \emptyset \\
        0  & \text{otherwise,}
    \end{cases}
\end{align} and the elements of $R$ are \begin{align}
    R_{(i,j),(k,\ell)} = \begin{cases}
        [(m-6)(m-5) + 4(m-4)](m-4)!  & \text{if } \{i,j\} \cap \{k,\ell\} = \emptyset \\
        (m-4)[(m-3)!]  & \text{if } i = \ell, j \neq k \text{ or } i \neq \ell, j = k \\
        0  & \text{otherwise.}
    \end{cases}
\end{align}
\end{theorem}

\subsection{The TE Model Under Block Constraints} 

We now consider a case where constraints are placed on the order-of-addition, and all $m!$ permutations of $(1,\dots,m)$ are no longer feasible experimental runs. Consider an experiment with $m$ components, where each of the components is placed into one of $c$ blocks labeled $b_1, \dots, b_c$ such that if $i,j \in \{1,\dots,c\}$ and $i < j$, then all components in block $b_i$ must come before all components in block $b_j$. This way, all components within a block may be arranged in any order, but the $c$ blocks have a fixed order.  


\begin{figure}[!h]
    \centering
    \includegraphics[scale = 0.5]{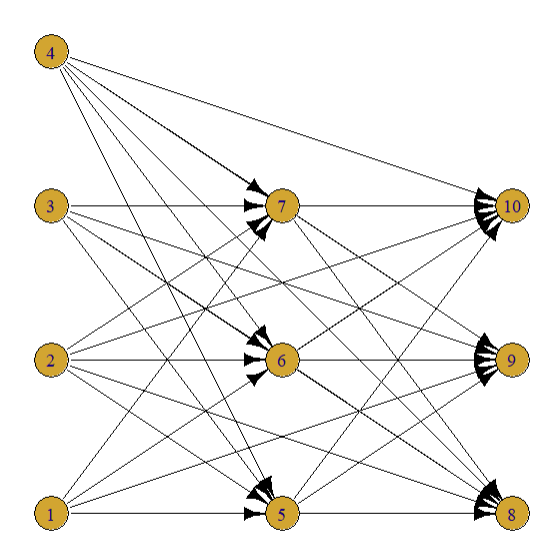}
    \caption{Example of Block Constraints}
    \label{fig:flowchart}
\end{figure} For example, in Figure \ref{fig:flowchart} there are ten components that are arranged into three blocks. Components 1,2,3, and 4 belong to block $b_1$, components 5,6, and 7 belong to block $b_2$, and the remaining components 8,9, and 10 belong to block $b_3$. All components in block 1 precede all components in block 2, and all components in block 2 precede all components in block 3. As Figure 1 shows, it is convenient to represent these constraints as a directed acyclic graph, which we denote as $G$. Equivalently, we may represent the set of constraints with a set of directed edges $B = \{(i,j) \mid i \text{ precedes } j \text{ in } G, \, \text{ where } i , j \in \{ 1,\dots,m \}  \}$.

When block constraints are present, it is mainly of interest to identify the optimal order within each block. This is because the order of the blocks is known, so the focus shifts to finding the order of the components in each block that is optimal. Therefore, we simplify the length 1 transition-effect model (\ref{eqn:oofamodel}) by only considering transition effects between components within the same block. This reduced model is \begin{align}
    \label{eqn:oofamodelreduced}
    \tau( \textbf{a}) =  \beta_0 + \sum_{i=1}^c\sum_{j=1}^m\sum_{k=1, k \neq j}^m I(j,k \in b_i)\beta_{j,k} x_{j,k}(\textbf{a}),
\end{align} where $I(j,k \in b_i) = 1$ if components $j$ and $k$ both belong to block $b_i$, and $I(j,k \in b_i) = 0$ otherwise. Notice that for any feasible permutation that satisfies the block constraints, $ \sum_{j=1}^m \sum_{k=1, k\neq j}^m I(j,k \in b_i) x_{j,k}(\textbf{a}) = m_i$, where $m_i$ is the number of components in block $b_i$. Since this would create a linear dependency, we set exactly one transition effect to zero per block to ensure identifiability. Similarly, the length 2 transition-effect model (\ref{eqn:oofamodel2}) can be reduced to \begin{align}
    \label{eqn:oofamodel2reduced}
    \tau(\textbf{a}) = \beta_0 + \sum_{i=1}^c\sum_{j=1}^m\sum_{k=1, k \neq j}^m I(j,k \in b_i)\Big[ \beta_{j,k}x_{j,k}(\textbf{a}) +  \delta_{j,k}x_{j,k}^{(2)}(\textbf{a})\Big],
\end{align} where $I(j,k \in b_i)$ is defined as it is for Model (\ref{eqn:oofamodelreduced}). To evaluate the $I-$ and $D-$ efficiencies of the full designs for these reduced models, it is important to be able to quickly construct the moment matrix without having to enumerate all possible permutations. Corollary \ref{cor:momentmat3} shows how this can be accomplished. \begin{Corollary}[Moment Matrix for the Full Design Under Model (\ref{eqn:oofamodelreduced})]
\label{cor:momentmat3}
Suppose that the $m$ components are arranged into $c$ blocks $b_1, \dots, b_c$ such that if $i < j$, then all components in $b_i$ must precede all components in $b_j$. Re-label the components so that components $1,\dots,n_1$ are in $b_1$, $n_{i-1} + 1, \dots, n_i$ are in block $b_i$ for $i = 2,\dots,c-1$, and $n_{c-1}+1, \dots, m$ are in $b_c$. Let $m_i$ be the number of components in block $b_i$. Let $D_m$ be the full design that uses each of the $N = \prod_{i=1}^c (m_i)! $ feasible permutations exactly once. Let $X_m$ be the model matrix expansion of $D_m$ under Model (\ref{eqn:oofamodelreduced}), with the constraint that $\beta_{m,m-1} = 0$. For $i = 1,\dots, c$,  let $q_i = 2{m_i \choose 2} - 1$, and let $V_i$ be a $q_i \times q_i$ matrix that is defined similarly to Theorem \ref{thm:momentmat1}. Then it follows that \begin{align}
    X_m^TX_m/N = \begin{bmatrix}
        1  & (1/m_1)\textbf{1}^T &  (1/m_2)\textbf{1}^T  & \dots & (1/m_c)\textbf{1}^T  \\
        (1/m_1)\textbf{1}  & M_1 & \frac{1}{m_1m_2}\textbf{1}_{q_1 \times q_2} & \dots & \frac{1}{m_1m_c}\textbf{1}_{q_1 \times q_c} \\
        (1/m_2)\textbf{1}  & \frac{1}{m_2m_1}\textbf{1}_{q_2 \times q_1}  & M_2 & \dots, & \frac{1}{m_2m_c}\textbf{1}_{q_2 \times q_c}\\
        \vdots & \vdots & \vdots & \ddots & \vdots \\
        (1/m_c)\textbf{1}  & \frac{1}{m_cm_1}\textbf{1}_{q_c \times q_1} & \frac{1}{m_cm_2}\textbf{1}_{q_c \times q_2} & \dots & M_c
    \end{bmatrix},
\end{align} where $\textbf{1}$ is a conformable column vector of 1s, $\textbf{1}_{q_i \times q_j}$ is a $q_i \times q_j$ matrix of ones, and  \begin{align*}
    M_i = [(m_i-1)!I_{q_i} + (m_i-2)!V_i]/N
\end{align*} for $i = 1,\dots,c$.
\end{Corollary}

\subsection{Finding Efficient Exact Designs}

Several approaches exist for finding efficient exact experimental designs under an optimality criterion. For small values of $m$, it is very fast to enumerate all $m!$ possible runs (or fewer, if block constraints are present), and one can use an exchange algorithm to find an efficient design. In this case, a convenient solution is to use \texttt{optFederov()} in R, which can be found in \texttt{AlgDesign} package \citet{algdesign2022}. However, for large $m$ (e.g. $m \geq 9$), it becomes costly to enumerate all $m!$ possible runs, and this approach can be inefficient. Therefore, we propose methods for finding exact designs with high efficiency that do not require enumerating all feasible orders of addition.

\subsubsection{Simulated Annealing}
 
An algorithm that can be used to find efficient designs is the simulated annealing algorithm \citet{kirkpatrick1983optimization}. Simulated annealing (SA) is a general method for finding the minimum (or maximum) of a function. In general, the algorithm starts with an initial solution, and then proposes a new solution somewhere in the neighborhood of the current solution. This new solution is accepted with a probability that depends on a ``temperature'' and the values of the function at the current and proposed solutions. The temperature is a function that decreases (or ``cools'') as the number of iterations increase. When the number of iterations is still low, the temperature is high, which gives sub-optimal solutions a higher probability of being accepted. As the temperature cools, the algorithm will eventually only select new solutions that strictly improve the optimality of the solution.

\begin{algorithm}[H]
\label{Alg:SA}
  \SetAlgoLined
  \textbf{Input}: A number of components $m$, set of constraints $B$, sample size $n$, optimality criterion $\phi$, and the number of iterations $n_i$.\\

  1. Randomly initialize $D^{(0)}$, generate $M^{(0)}$, and calculate $I^{(0)} = \phi(M^{(0)})$. Set $I^* = I^{(0)}$. \\
  \For{$t = 1,2,\dots,n_i$}{
     \While{A new design has not been found}{
     2. Randomly select a row $i$ and column $j$ of $D^{(0)}$, where $j \in \{1,2,\dots,m-1\}$, and exchange elements in positions $j$ and $j+1$ in the $i^{th}$ row of $D^{(0)}$ to get a new matrix $D^{(1)}$. \\ 
     \If{the constraints in $B$ are not violated}{
        3. Keep $D^{(1)}$ and exit the loop.}
    }
    4. Find $M^{(1)}$ and store $I^{(1)} = \phi(M^{(1)})$. \\
    \If{  $U < \exp{ \{- [I^{(1)} - I^{(0)}]/[1/\log(t+1)]\}}$  }{
         5. Update $D^{(0)} = D^{(1)}$, $M^{(0)} = M^{(1)}$, $I^{(0)}=I^{(1)}$. \\
         6. If $I^{(0)} < I^*$, then update $D^* = D^{(0)}$ and $I^* = I^{(0)}$. 
    }
    }
 \Return{$D^*$}
 \caption{Simulated Annealing with Block Constraints}
\end{algorithm}

Algorithm \ref{Alg:SA} uses simulated annealing to generate designs that try to minimize an optimality criterion $\phi(\cdot)$. For $I-$optimality, $\phi(M) = \text{trace}(M^{-1}M_f)$, and for $D-$optimality, $\phi(M) = -\text{log}(|M|)$. In the unconstrained case, $M_f$ can be quickly found using Theorem (\ref{thm:momentmat1}) for Model (\ref{eqn:oofamodel}), or Theorem (\ref{thm:momentmat2}) for Model (\ref{eqn:oofamodel2}), so complete enumeration of the full design is not required to evaluate the $I-$efficiency of the design. If block constraints are present, Corollary \ref{cor:momentmat3} can be used to do this. At each iteration, a new design matrix is proposed by swapping adjacent components, provided that this does not violate the pairwise constraints. Instead of automatically adopting the new design, a new design is accepted with probability $\exp{ \{- [I^{(1)} - I^{(0)}]/[1/\log(t+1)]\}}$, where $I^{(0)}$ is the $\phi-$optimality of the previous design, $I^{(1)}$ is the $\phi-$optimality of the proposed design, and $t$ is the current iteration number. This means that it is possible to accept designs that have worse $\phi-$optimality earlier on in the algorithm, which is helpful for exploring the space of possible designs and avoiding becoming stuck at a locally optimal design. As $t$ increases, the algorithm will focus on finding more optimal designs.

\subsubsection{Bubble Sorting Algorithm}

Another possible algorithm for finding efficient fixed designs is to use a sorting algorithm. The general idea is to iteratively ``sort'' each row of the design to find the order of the components in each row that maximize the optimality criteiron of the entire design. In \citet{lin2019order}, the bubble sort algorithm was used to find efficient OofA designs under the $D-$optimality criterion. This procedure is summarized in Algorithm \ref{Alg:BS} below, which is adapted to handle block constraints. 

\begin{algorithm}[H]
\label{Alg:BS}
  \SetAlgoLined
  \textbf{Input}: A number of components $m$, set of constraints $B$, sample size $n$, optimality criterion $\phi$, and the number of maximum iterations $iter_{max}$.\\
  1. Randomly initialize $D^{(0)}$ with $n$ rows, and generate $M^{(0)}$. Calculate $I^{(0)} = \phi(M^{(0)})$. \\
  \For{$iter = 1,2,\dots,iter_{max}$}{
     \For{$r = 1,2,\dots,n$}{
        2. \texttt{Still\_Sorting} = True. \\
        \While{\texttt{Still\_Sorting}}{
        \For{$j = i,\dots,m-1$}{
            3. Select the $r^{th}$ row of $D^{(0)}$, then check if the elements in positions $j$ and $j+1$ can be exchanged to get a new matrix $D^{(1)}$.\\
            \If{the constraints in $B$ are not violated}{
            4. Find $M^{(1)}$ and store $I^{(1)} = \phi(M^{(1)})$.\\
             \eIf{  $I^{(1)}<I^{(0)}$  }{
         5. $D^{(0)} = D^{(1)}$, $M^{(0)} = M^{(1)}$, $I^{(0)}=I^{(1)}$.  \texttt{Still\_Sorting} = True. \\
    }{
    6. \texttt{Still\_Sorting} = False. 
    }}
            
        } }
    }
    
     }
    8. Store the best design so far as $D^*=D^{(0)}$ and the corresponding efficiency $I^*=I^{(0)}$.\\
 \Return{$D^*$}
 \caption{Bubble Sort with Block Constraints}
\end{algorithm}

As an example, suppose we have a design matrix with $n$ rows for $m = 4$ components, and the first row of it is $(1,3,2,4)$. The bubble sort algorithm will first try to swap 1 and 3. It will compare the $\phi-$efficiency of the design with 1,3 to the design with 3,1. If the new design is better, 1 and 3 will be swapped. Suppose they are swapped, so the first row is now $(3,1,2,4)$. Then, the algorithm will try switching 2 and 1. The algorithm will continue along this row until it is unable to swap two adjacent elements in a way that improves the design efficiency. Once the first row is ``sorted'' (i.e., the best ordering is found for the first row), the procedure continues for the following rows of the design matrix. This process is repeated $iter_{max}$ times for the entire matrix. Overall, Algorithm \ref{Alg:BS} is a greedy algorithm, since it only performs exchanges that strictly improve the $\phi-$optimality criterion. 

\subsubsection{Greedy Randomized Adaptive Search Procedure (GRASP)}


Another algorithm that can be used to find efficient designs is the Greedy Randomized Adaptive Search Procedure (GRASP). GRASP is a metaheuristic that iterates two main steps: first, a randomized greedy construction method is used to construct an initial solution, and then, a local optimization procedure is used to improve the solution \citet{resende2010greedy}. This procedure is repeated several times, and the best solution is returned at the end. In this case, a ``solution'' is a design with $n$ runs. The randomized greedy construction method constructs a random design with $n$ runs and a nonsingular moment matrix. It then considers several pairwise exchanges, and randomly selects one of the $q$ best exchanges, for some natural number $q$. This way, more varied initial design are proposed. Our implementation of the GRASP algorithm is summarized in Algorithm \ref{Alg:GRASP}.  

\begin{algorithm}[H]
\label{Alg:GRASP}
  \SetAlgoLined
  \textbf{Input}: A number of components $m$, set of constraints $B$, sample size $n$, optimality criterion $\phi$, number of candidates $n_{cand}$, and the number of iterations $n_i$.\\

  \For{$t = 1,2,\dots,n_i$}{
      1. Randomly find a design $D^{(0)}$ with $n$ runs and non-singular moment matrix $M$, and calculate $I = \phi(M)$. Set $q = 10$. Initialize a list $L$ of length $n_{cand}.$ \\
      \For{$i = 1,2,\dots,n_{cand}$}{
            2. Randomly select a row $r_i$ and column $j_i$ of $D$ to exchange, where $j \in \{1,\dots,m-1\}$, and exchange elements in positions $j_i$ and $j_i+1$ of the $r^{th}$ row of $D^{(0)}$ to get $D^{(i)}$. \\
            \If{the constraints in $B$ are not violated}{
                3. Find $M^{(i)}$ and store $\phi(M^{(i)})$ in the $i^{th}$ position of $L$. \\
                4. Store the values of $r_i$ and $j_i$ used to construct $D^{(i)}$. \\
            } 
      }
    5. Let $I_q$ be the indices of the $q^{th}$ smallest elements of $L$. \\
    6. Randomly select an index $\ell$ from $I_q$. Let $D^{(\ell)}$ be the design found from exchanging the elements in positions $j_{\ell}$ and $j_{\ell} + 1$ in row $r_{\ell}$ of $D^{(0)}$. \\
    7. Use a local search procedure with $D^{(\ell)}$ as the initial design to find a locally optimal design and its $\phi-$criterion. Store this design as the new $D^{(0)}$. When $t$ is a multiple of 10, update $q = \text{max}(q-1,1)$.   \\
    8. Store the design with the lowest $\phi-$criterion so far as $D^*$.
    }
 \Return{$D^*$}
 \caption{GRASP with Block Constraints}
\end{algorithm}

Algorithm \ref{Alg:GRASP} first finds a random design $D$ with $n$ runs with a non-singular moment matrix $M$. In Steps 2-4, $n_{cand}$ possible exchanges are considered, each resulting in $n_{cand}$ possible designs. The $\phi-$criteria are found for each candidate design, and are stored in a list $L$. In Step 5, the designs with the $q$ smallest $\phi-$criteria are identified (i.e. the best $q$ designs), and one of the best $q$ designs is randomly selected in Step 6.  Steps 1 to 6 comprise the ``greedy randomized construction'' procedure that is common to GRASP algorithms \citet{resende2010greedy}. The value of $q$ (initialized at 10) determines the quality of designs that the algorithm is willing to consider for the initial design $D$. Smaller values of $q$ mean that the GRASP algorithm will only consider exchanges that have better efficiency, which correspond to lower $\phi-$ values in this case. When $q = 1$, the algorithm is fully greedy in the sense that it only considers the best exchanges in each iteration. After an initial design is constructed, it is passed to a local search procedure in Step 7, and the best $\phi-$optimality is stored in Step 8. The role of the local search procedure is to find the best design among a set of neighboring designs. In our implementation, the local search procedure randomly selects one pair of adjacent components in each row, and swaps them if it would result in a more optimal design. Steps 1 to 8 repeat $n_i$ times. Intuitively, higher values of $n_{cand}$ and $n_i$ lead to more efficient designs, since this would allow for more design search attempts. In our implementation, we found that using $n_{cand} = 250$ and $n_i = 100$ were sufficient for identifying highly efficienct exact designs.

\section{Empirical Results}
\label{sec:simuresults}
In this section, the algorithms introduced in Section \ref{sec:methods}.3 are used to construct designs under various scenarios. The relative $D-$ and $I-$ efficiencies of these designs are compared to that of the full design, which uses all feasible permutations. The goal is to see if efficient, cheaper designs could be found in a variety of different settings and under various constraints on the pairwise order of the components. These comparisons are made with no constraints in Section \ref{sec:simuresults}.1 and with block constraints in Section \ref{sec:simuresults}.2. In all scenarios, comparisons are made for relatively large numbers of components $m \geq 9$. 

\subsection{Unconstrained Relative Efficiencies}

The Simulated Annealing, Bubble Sort, and GRASP algoirthms were used to find $D-$ and $I-$ efficient designs for $m = 9, 10, 11$. These designs were found for varying sample sizes, denoted $n_s = 400, 500, 600$. These sizes were chosen because they are a small proportion of the $m!$ available runs, and in each case the sample size is larger than $2{m \choose 2} + 1$, which is sufficient for estimating all parameters in Model (\ref{eqn:oofamodel}). Each algorithm was executed 20 times, and the median relative $D-$ and $I-$efficiencies were the recorded in Tables \ref{tab:intermediansD} and \ref{tab:intermediansI}, respectively.
 
\begin{table}[!ht]
    \centering
    \begin{tabular}{||c c c c c c||} 
 \hline
 m & $n_s$ & $n_s/m!$ & Simulated Annealing & Bubble Sort & GRASP \\ [0.5ex] 
 \hline\hline
 9 & 400 & 0.0011 & 0.9111  & 0.9723  & 0.9822 \\ [1ex] 
 \hline
  & 500 & 0.0014 & 0.9277 & 0.9809  & 0.9885 \\ [1ex] 
 \hline
  & 600 & 0.0017 & 0.9403 & 0.9861 & 0.9919 \\ [1ex] 
 \hline
 10 & 400 & 0.0001 & 0.8856 & 0.9565 & 0.9725\\ [1ex] 
 \hline
  & 500 & 0.0001 & 0.9098 & 0.9699  & 0.9795 \\ [1ex] 
 \hline
  & 600 & 0.0002 & 0.9253 & 0.9773  & 0.9853 \\ [1ex] 
 \hline
  11 & 400 & $<0.0001$ & 0.8579  & 0.9368 & 0.9430 \\ [1ex] 
 \hline
  & 500 & $<0.0001$ & 0.8880 & 0.9554 & 0.9608 \\ [1ex] 
 \hline
  & 600 & $<0.0001$ & 0.9075 & 0.9663 & 0.9707 \\ [1ex] 
 \hline

\end{tabular}
    \caption{Median Relative $D-$ Efficiencies}
    \label{tab:intermediansD}
\end{table}

\begin{table}[!ht]
    \centering
    \begin{tabular}{||c c c c c c||} 
 \hline
 m & $n_s$ & $n_s/m!$ & Simulated Annealing & Bubble Sort & GRASP \\ [0.5ex] 
 \hline\hline
 9 & 400 & 0.0011 & 0.8229 & 0.9435  & 0.9642  \\ [1ex] 
 \hline
  & 500 & 0.0014 & 0.8576  & 0.9619  & 0.9769 \\ [1ex] 
 \hline
  & 600 & 0.0017 & 0.8805 & 0.9725 & 0.9841 \\ [1ex] 
 \hline
 10 & 400 & 0.0001 & 0.7771 & 0.9114 & 0.9376 \\ [1ex] 
 \hline
  & 500 & 0.0001 & 0.8211 & 0.9386 & 0.9583 \\ [1ex] 
 \hline
  & 600 & 0.0002 & 0.8508 & 0.9551 & 0.9705 \\ [1ex] 
 \hline
  11 & 400 & $<0.0001$ & 0.7243  & 0.8738 &  0.9021\\ [1ex] 
 \hline
  & 500 & $<0.0001$ & 0.7796 & 0.9089 & 0.9332 \\ [1ex] 
 \hline
  & 600 & $<0.0001$ & 0.8172 & 0.9319 & 0.9517 \\ [1ex] 
 \hline

\end{tabular}
    \caption{Median Relative $I-$ Efficiencies}
    \label{tab:intermediansI}
\end{table}

Tables \ref{tab:intermediansD} and \ref{tab:intermediansI} show the median relative $D-$ and $I-$ efficiencies, respectively, of designs of size $n_s$ to the full design for the simulated annealing (SA), bubble sort, and GRASP algorithms. For all methods, as the subset size $n_s$ increases, the median relative $D-$ and $I-$ efficiency increases. It can be seen that as $m$ increases, GRASP is the preferred method in terms of relative $D-$ and $I-$efficiency. The median relative efficiences of all methods tends to decrease as $m$ increases; as $m$ increases, the set of all possible designs grows much larger, so it is more difficult to search for optimal designs. In particular, the SA algorithm struggles to find designs with high $I-$efficiency for $m = 11$ and $n_s = 400, 500$. Both GRASP and bubble sort consistently outperform SA in terms of the $D-$ and $I-$optimality criteria. Furthermore, in every case, the designs found using GRASP have higher median relative $D-$ and $I-$efficiencies than those found using bubble sort. 

\subsection{Constrained Relative Efficiencies}

The Simulated Annealing, bubble sort, and GRASP algorithms were also used to find $D-$ and $I-$efficient designs under block constraints for $m = 9,10,11$. For each value of $m$, the number of blocks $c$ was varied, and values of $c = 2,3$ were considered. This was done to examine the effect of the number of blocks on the relative $D-$ and $I-$efficiencies of the designs. In each scenario, the number of components were placed as evenly as possible among the $c$ blocks. If $m$ was not evenly divisible by $c$, extra components were placed in the latter blocks. Tables \ref{tab:constrainedDeffs} and \ref{tab:constrainedIeffs} specifically show the distribution of the components among the $c$ blocks ($m_1, \dots, m_c$) for each scenario. When $m = 9$, designs of size $n_s = 50, 100$ were found; for $m \geq 10$, designs of size $n_s = 100,200$ were found. The ratio of the sample size $n_s$ to the size of the full design under block constraints ($N = \prod_{i=1}^c m_i!$) is also shown in Tables \ref{tab:constrainedDeffs} and \ref{tab:constrainedIeffs}.

\begin{table}[!ht]
\centering 
\begin{tabular}{||llllllll||}
\hline 
m  & c & $(m_1, \dots, m_c)$                           & $n_s$ & $n_s/N$ & SA & Bubble Sort & GRASP    \\ [0.5ex]
\hline \hline 
9  & 2 & $(4,5)$           & 50   & 0.0174 & 0.9185 & 0.8737 & 0.9224   \\ [1ex] \hline 
   &   &                                & 100  & 0.0347 & 0.9919 & 0.9719 & 0.9902  \\ [1ex] \hline  
   & 3 & $(3,3,3)$ & 50   & 0.2314 & 0.9857 & 0.9808 & 0.9855    \\ [1ex] \hline  
   &   &                                & 100  & 0.4630 & 0.9964 & 0.9957 & 0.9964  \\ [1ex] \hline 
10 & 2 & $(5,5)$           & 100   & 0.0069 &  0.9712 & 0.9446  & 0.9670   \\ [1ex] \hline 
   &   &                                & 200  & 0.0138 & 1.0002 & 0.9899 & 0.9988   \\ [1ex] \hline 
   & 3 & $(3,3,4)$ & 100   & 0.1157 & 1.0038 & 0.9974  & 1.0033    \\ [1ex] \hline 
   &   &                                & 200  & 0.2314 & 1.0085 & 1.0069 & 1.0083   \\ [1ex] \hline 
11 & 2 & $(5,6)$           & 100   & 0.0012 & 0.9375  & 0.8999  & 0.9299    \\ [1ex] \hline 
   &   &                                & 200  & 0.0024 & 0.9899 & 0.9740  & 0.9880 \\ [1ex] \hline 
   & 3 & $(3,4,4)$ & 100   & 0.0289 & 0.9961 & 0.9842 & 0.9944   \\ [1ex] \hline 
   &   &                                & 200  & 0.0578 & 1.0051 & 1.0016 & 1.0045  \\ [1ex] \hline \hline 
\end{tabular}
\caption{Median Relative $D-$Efficiency Comparison Under Block Constraints}
\label{tab:constrainedDeffs}
\end{table}

Tables \ref{tab:constrainedDeffs} and \ref{tab:constrainedIeffs} show the median relative $D-$ and $I-$efficiencies, respectively, of designs obtained by the SA, Bubble Sort, and GRASP algorithms to the full design. In both tables, increasing the sample size generally increases the median relative efficiency of the design with respect to the full design. All three algorithms typically found designs with higher relative $D-$efficiency than relative $I-$efficiency, especially in the case when $m = 9, n = 50$. Of the three methods used to construct designs, the Bubble Sort algorithm had designs with the lowest relative $D-$ or $I-$efficiencies in most cases. In Table \ref{tab:constrainedDeffs}, the SA algorithm typically performs the best, and it is closely followed by the GRASP algorithm. This is very different from the unconstrained scenarios in Section \ref{sec:simuresults}.1, where the SA algorithm did not perform well; this difference likely occurs because the presence of constraints reduces the number of possible permutations that are feasible in a design.

\begin{table}[!ht]
\centering 
\begin{tabular}{||llllllll||}
\hline 
m  & c & $(m_1, \dots, m_c)$                           & $n_s$ & $n_s/N$ & SA & Bubble Sort & GRASP    \\ [0.5ex]
\hline \hline 
9  & 2 & $(4,5)$           & 50   & 0.0174 & 0.8425 & 0.7269  & 0.8256   \\ [1ex] \hline 
   &   &                                & 100  & 0.0347 & 0.9419 & 0.9311 & 0.9652   \\ [1ex] \hline  
   & 3 & $(3,3,3)$ & 50   & 0.2314 & 0.9710 & 0.9608 &  0.9711   \\ [1ex] \hline  
   &   &                                & 100  & 0.4630 & 0.9676 & 0.9909 &  0.9928  \\ [1ex] \hline 
10 & 2 & $(5,5)$           & 100   & 0.0069 & 0.9149  & 0.8712 & 0.9219   \\ [1ex] \hline 
   &   &                                & 200  & 0.0138 & 0.8764 & 0.9685 & 0.9848   \\ [1ex] \hline 
   & 3 & $(3,3,4)$ & 100   & 0.1157 & 0.9585 & 0.9807 & 0.9913    \\ [1ex] \hline 
   &   &                                & 200  & 0.2314 & 0.9351 & 0.9995 & 1.0020    \\ [1ex] \hline 
11 & 2 & $(5,6)$           & 100   & 0.0012 & 0.8689  & 0.7822 & 0.8530   \\ [1ex] \hline 
   &   &                                & 200  & 0.0024 & 0.8475  & 0.9388 & 0.9658 \\ [1ex] \hline 
   & 3 & $(3,4,4)$ & 100   & 0.0289 & 0.9479 & 0.9570 & 0.9778   \\ [1ex] \hline 
   &   &                                & 200  & 0.0578 & 0.9169 & 0.9922 & 0.9976  \\ [1ex] \hline \hline 
\end{tabular}
\caption{Median Relative $I-$Efficiency Comparison Under Block Constraints}
\label{tab:constrainedIeffs}
\end{table}

Table \ref{tab:constrainedIeffs} shows that in all but two cases, the GRASP algorithm produces designs with higher relative $I-$efficiency than SA and Bubble Sort. The two exceptions are when $m = 9, n = 50, c = 2$ and $m = 11, n = 100, c = 2$; in both cases, the GRASP algorithm was slightly behind the SA algorithm. Several entries in Table \ref{tab:constrainedDeffs} have relative efficiencies greater than 1, and so does one entry in Table \ref{tab:constrainedIeffs}. This shows that, in general, when block constraints are applied, the full design is not necessarily optimal. 

\section{Example}
\label{sec:example}

In this section, an example from \citet{skorobohatyj} is used to study the performance of the Bubble Sort (BB) and GRASP algorithms and to compare the performance of the TE and PWO models at detecting the optimal order of addition. This data set was created using IBM data and used for the Sequential Ordering Problem, which was also examined in \citet{karan2011branch}. In this problem, there are $m = 11$ components in total. The original problem allows for all $11!$ possible orders to be examined. To highlight our method's ability to handle block constraints, we randomly divided the components into 3 blocks, which are $b_1 = \{1,4,7,9,10\}$, $b_2 = \{2,6,8\}$, and $b_3 = \{3,5,11\}$. In this case, all components in $b_1$ precede all components in $b_2$, and all components in $b_2$ precede all components in $b_3$. There are $N = 5!3!3! = 4320$ possible orders of addition under these constraints. For each row, the response variable $y$ was found, which represents the cost of completing all 11 jobs in a particular order. The order with the lowest cost is optimal. 



The BB and GRASP are used to obtain optimal designs of size $n = 150$ and $n = 200$, which represent roughly $3.5\%$ and $4.6\%$ of the total number of available runs, respectively. These algorithms were used to search for $I-$optimal designs under the length-one TE model (\ref{eqn:oofamodelreduced}). These designs were used to fit the length-one TE model, the length-two TE model (\ref{eqn:oofamodel2reduced}), and the PWO model. Once fit, each of the three models was used to find the order with the lowest estimated cost. This procedure was repeated 100 times. The average rank over all 100 attempts was used as a metric to compare the performance of the methods. The results are given in Table \ref{tab:case1}.\\ 




\begin{table}[!ht]
\centering 
\begin{tabular}{||ccccc||}
\hline \hline 
    & \multicolumn{2}{c}{BB} & \multicolumn{2}{c||}{GRASP} \\
    \hline 
n   & 150        & 200       & 150         & 200         \\
\hline 
PWO & 54.78      & 46.39     & 45.58       & 56.73       \\
\hline
TE1 & 18.58      & 13.11     & 14.91       & 14.62       \\
\hline 
TE2 & 8.50      & 12.61     & 10.60       & 18.10  \\    
\hline  \hline 
\end{tabular}
\caption{Average rank of the estimated optimal order. Lower ranks are more optimal.}
\label{tab:case1}
\end{table}

Table \ref{tab:case1} shows the average of the optimal order's position in 100 prediction attempts. From the table, we can notice that both TE models outperform PWO model under both design algorithms and for both sample sizes. When $n = 150$, the length-two TE model (TE2) is only slightly better than the length-one TE model (TE1), while it achieves a better result with the BB algorithm when $n=200$. However, the length-one TE model is better with the GRASP algorithm compared to the length-two TE model when $n=200$. The two design algorithms have similar performance with the length-one TE model when $n=200$, and bubble sorting outperformed GRASP when the PWO model was used when $n=150$. Furthermore, it is noticed that if $n$ becomes even larger, the difference in average ranking between BB and GRASP diminishes under the length-one TE model, though the gap between PWO and TE remains. Overall, TE1 is more stable than TE2, and its results improve as the sample size increases. This example shows that, with an appropriately selected design, the TE model can identify orders with lower cost more frequently than the PWO model. 


\section{Conclusion and Further Research Interests}
\label{sec:conc}


This article introduced a transition effect model for the Order-of-Addition (OofA) problem. This model represents the response using transitions between nearby components in the model. Initial theory shows that the full design is $\phi-$optimal for the proposed transition effect models under many criteria, and closed-form constructions for the moment matrix of the full design were provided. These initial results provide a fast way to evaluate the $D-$ and $I-$optimality criterion of any design relative to the full design. A novel implementation of a Greedy Randomized Adaptive Search Procedure (GRASP) was implemented to search for cost-efficient designs with high efficiency relative to the full design. This algorithm was compared with Simulated Annealing (SA) and Bubble Sort (BB), which are two alternative approaches for constructing highly efficient OofA designs. In Section \ref{sec:simuresults}, empirical results showed that GRASP and Bubble Sorting both outperformed the SA algorithm in terms of finding efficient designs under the $D-$ and the $I-$optimality criteria. In Section \ref{sec:example}, the proposed TE model was shown to be better than the PWO model in terms of average rank of predicted optimal order. Furthermore, the GRASP algorithm outperformed bubble sort when the sample size was smaller. 

There is more work that can be done to extend this research. One research direction would be to focus on a more general case of constraints. For example, what would happen if constraints exist on the pairwise order of components, but the components are not divided into even blocks? In Section \ref{sec:simuresults}, it was observed that the full design was not $D-$ or $I-$optimal in some cases. The theoretical results in this paper show $D-$optimality for the full design in the unconstrained case, but this does not generalize to cases where pairwise constraints exist on the order of addition. It would also be useful to develop theory for the closed-form construction of optimal fractional OofA designs under a transition-effect model. Additional avenues of research would be developing transition-effect that include other covariates in addition to the order of addition, such as mixture proportions \citet{rios2021order} or factorial effects \citet{tsai2023dual}. 

\section*{Data availability statement}
The authors confirm that the data supporting the findings of this study are available within the article and its supplementary materials.

\section*{Funding}

This research was not supported by any grants or funding. 

\section*{Disclosure statement}
No potential conflict of interest was reported by the authors. 

\bibliographystyle{chicago}
\bibliography{biblio.bib}

\section*{Appendix - Proofs}

\begin{Lemma}
\label{lemma:permmatrix}

Let $\textbf{a} \in D_m$. Let $x(\textbf{a})$ be the model matrix expansion of $\textbf{a}$ under a length $\ell$ transition effect model, where $\ell \in \{1,2\}$. Let $\pi \in D_m$. For a permutation $\textbf{a} \in D_m$, let $\pi\textbf{a} = (\pi_{a_1},\pi_{a_2}, \pi_{a_3} \dots,\pi_{a_m})$. Then there is a permutation matrix $R_{\pi}$ such that $x(\pi \textbf{a})^T = x(\textbf{a})^TR_{\pi}$.
\end{Lemma}

\begin{proof}
Let $\pi \in D_m$ and $\textbf{a} \in D_m$. Let $j$ and $k$ be distinct elements of $\{1,2,\dots,m\}$. If $d(j,k,\textbf{a}) = \ell$, then by definition of $\pi \textbf{a}$, $d(\pi_{a_j}, \pi_{a_k}, \pi \textbf{a}) = \ell$. Hence, it follows that $x_{jk}(\textbf{a}) = x_{\pi_j,\pi_k}(\pi\textbf{a})$ for any pair $jk$.

If $\ell = 1$, let $R_{\pi}^{(1)}$ be a matrix with columns and rows indexed by the pairs $(12,13,...(m-1)m)$ such that $R_{\pi}^{(1)}(jk,\pi_j,\pi_k) = 1$ for each pair $jk$ and the remaining elements are 0. Then $x(\pi \textbf{a})^T = x(\textbf{a})^T\text{diag}(1,R_{\pi}^{(1)})$ for any $\pi \in D_m$. So, $R_{\pi} = \text{diag}(1,R_{\pi}^{(1)})$ is a signed permutation matrix which satisfies the lemma.

If $\ell = 2$, let \begin{align}
    R_{\pi} = \begin{bmatrix}
        1  & \textbf{0}  & \textbf{0} \\
        \textbf{0} & R_{\pi}^{(1)}  & \textbf{0}\\
        \textbf{0} & \textbf{0}   &  R_{\pi}^{(2)}
    \end{bmatrix}
\end{align} where $R_{\pi}^{(2)}$ is a matrix with columns and rows indexed by the pairs $(12,13,\dots,(m-1)m)$ such that $R_{\pi}^{(2)}(jk,\pi_j,\pi_k) = 1$ for each pair $jk$ and the remaining elements are 0, and $\textbf{0}$ is a conformable matrix of zeroes. Then $x(\pi \textbf{a})^T = x(\textbf{a})^TR_{\pi}$ for any $\pi \in D_m$.  
\end{proof}

\noindent \textbf{Proof of Corollary \ref{cor:fullopt}.}

\begin{proof}
    The structure of this proof is very similar to that of Theorem 1 in \citet{OOADesign} (which is why we denote this as a Corollary). We write the proof in its entirety for clarity and convenience.     

    Let $D_m$ be the set of all $m!$ permutations of $(1,2,\dots,m)$. Let $\phi$ be a criterion that is concave and permutation invariant. As in the proof of Lemma \ref{lemma:permmatrix}, let $x(\textbf{a})$ represent the model matrix expansion of $\textbf{a}$ under a transition-effect model of length 1 or 2. Then $M_f = (1/m!)\sum_{\textbf{a} \in D_m} x(\textbf{a})x(\textbf{a})^T$.
    
    Let $w$ be an arbitrary design measure over $D_m$. For $\pi \in D_m$, let $\pi\textbf{a} = (\pi_{a_1}, \dots, \pi_{a_m})$, and let $\pi w$ be the design measure that assigns, for each $\textbf{a} \in D_m$, weight $w(\textbf{a})$ to the circuit $\pi \textbf{a}$. By concavity of $\phi$, we have that \begin{align}
    \label{ineq:concave}
    \phi\Big( \sum_{\pi \in D_m} \frac{1}{m!} M(\pi w) \Big) \geq \sum_{\pi \in D_m } \frac{1}{m!} \phi\Big( M(\pi w) \Big) 
    \end{align} Notice that, for a fixed $\textbf{a} \in D_m$, $\{ \pi\textbf{a} \mid \pi \in D_m  \} = D_m,$ which implies \begin{align}
     &\frac{1}{m!} \sum_{\pi \in D_m} M(\pi w) =  \frac{1}{m!}\sum_{\textbf{a} \in D_m} \sum_{\pi \in D_m} w(\pi \textbf{a}) x(\pi\textbf{a})x(\pi\textbf{a})^T  \\
      \label{eqn:LHS}
     = & \frac{m!}{m!} \sum_{\textbf{a} \in D_m} w(\textbf{a}) M_f = M_f.
    \end{align} Also notice that, by Lemma 1, \begin{align}
    &\phi\Big( M(w)  \Big) = \phi\Big( \sum_{\pi \in D_m} w(\pi \textbf{a}) x(\pi\textbf{a})x(\pi\textbf{a})^T \Big) = \\
    \label{eqn:RHS}
    &\phi(R_{\pi}^TM(w)R_{\pi}) = \phi(M(w))
    \end{align} where the last equality follows because $\phi$ is permutation-invariant. By substituting (\ref{eqn:LHS}) and (\ref{eqn:RHS}) into the inequality (\ref{ineq:concave}), we find that \begin{align}
        \phi\Big(  M_f \Big) \geq 
 \frac{1}{m!}\sum_{\pi \in D_m }  \phi\Big( M(w) \Big)  = \phi\Big( M(w) \Big),
    \end{align} for any design measure $w$ over $D_m$. This concludes the proof.

\end{proof}

\begin{proof}[\textbf{Proof of Theorem \ref{thm:momentmat1}}]
Let $D_m$ be the full design that uses all $m!$ permutations exactly once. Let $X_m$ be the model matrix expansion of $D_m$ under Model (\ref{eqn:oofamodel}), with the constraint that $\beta_{m,m-1} = 0$. Let $X_0 = [1,\dots,1]^T$ denote the column of $X_m$ that corresponds to the intercept $\beta_0$. Let $X_{j,k}$ denote the column of $X_m$ that corresponds to the transition effect $\beta_{j,k}$. Let $q = 2{m \choose 2} - 1$. Let $i,j,k,\ell$ be four distinct elements of $\{1,\dots,m\}$. Then, the following are true: \begin{enumerate}
    \item[(a)] $X_0^TX_0 = \sum_{\textbf{a} \in D_m} 1 = m!$
    \item[(b)] $X_0^TX_{j,k} = \sum_{\textbf{a} \in D_m} x_{j,k}(\textbf{a}) = (m-1)!$. The last equality is true because all permutations of $(1,\dots,m)$ where $j$ directly precedes $k$ can be constructed by treating $(j,k)$ as a single component and finding the set of all $(m-1)!$ permutations of the $(m-1)$ components $\{(j,k), a_1, \dots, a_{m-2}\}$, where $a_s \in \{1,\dots,m\}\setminus\{j,k\}, s = 1,\dots,m-2$. 
    \item[(c)] $X_{i,j}^TX_{i,k} = \sum_{\textbf{a} \in D_m} x_{i,j}(\textbf{a})x_{i,k}(\textbf{a}) = 0$ because for any $\textbf{a} \in D_m$, if $i$ directly precedes $j$, then $i$ cannot directly precede $k$; similarly, if $i$ directly precedes $k$, then $i$ cannot directly precede $j$. Therefore, no permutation of $(1,\dots,m)$ will simultaneously yield $x_{i,j}(\textbf{a}) = 1$ and $x_{i,k}(\textbf{a}) = 1$. 
    \item[(d)] $X_{i,j}^TX_{k,j} = \sum_{\textbf{a} \in D_m} x_{i,j}(\textbf{a})x_{k,j}(\textbf{a}) = 0$ because for any $\textbf{a} \in D_m$, if $i$ directly precedes $j$, then $k$ cannot directly precede $j$; similarly, if $k$ directly precedes $j$, then $i$ cannot directly precede $j$. Therefore, no permutation of $(1,\dots,m)$ will simultaneously yield $x_{i,j}(\textbf{a}) = 1$ and $x_{k,j}(\textbf{a}) = 1$. 
    \item[(e)]  $X_{i,j}^TX_{j,i} = \sum_{\textbf{a} \in D_m} x_{i,j}(\textbf{a})x_{j,i}(\textbf{a}) = 0$, because for any $\textbf{a} \in D_m$, $i$ cannot be both directly before and directly after $j$, so there is no permutation that simulatenously satisfies $x_{i,j}(\textbf{a}) = 1$ and $x_{j,i}(\textbf{a}) = 1$. 
    \item[(f)] $X_{i,j}^TX_{k,\ell} = \sum_{\textbf{a} \in D_m} x_{i,j}(\textbf{a})x_{k,\ell}(\textbf{a}) = (m-2)!$. This is true because there are exactly $(m-2)!$ permutations of $(1,2,\dots,m)$ where component $i$ is placed directly before component $j$ and (simultaneously) component $k$ is placed directly before component $\ell$. To see this, notice that the set of all permutations of $(1,\dots,m)$ that place $i$ directly before $j$ and $k$ directly before $\ell$ can be constructed by treating $i$ and $j$ as a single component $(i,j)$, treating $k$ and $\ell$ as a single component $(k,\ell)$, and then finding all permutations of the $m-2$ components $\{(i,j), (k,\ell), a_1, \dots, a_{m-4}\}$, where $a_s \in \{1,\dots,m\}\setminus\{i,j,k,\ell\}$ for $s = 1,\dots,m-4$.
    \item[(g)] $X_{i,j}^TX_{i,j} = \sum_{\textbf{a} \in D_f} [x_{i,j}(\textbf{a})]^2 = \sum_{\textbf{a} \in D_f} x_{i,j}(\textbf{a}) = (m-1)!$. The last equality follows from case (b).
\end{enumerate} Let $V$, $\textbf{1}$, and $I_q$ be as defined in Theorem \ref{thm:momentmat1}. Then, the results from (a) through (g) imply that \begin{align}
    M = X_m^TX_m/m! &= \frac{1}{m!}\begin{bmatrix}
        X_0^TX_0  & X_0^TX_{1,2} & \dots & X_0^TX_{m,m-2} \\
        X_{1,2}^TX_0 & X_{1,2}^TX_{1,2} & \dots & X_{1,2}^TX_{m,m-2} \\
        \vdots & \vdots & \ddots & \vdots \\
        X_{m,m-2}^TX_0 & X_{m,m-2}^TX_{1,2} & \dots & X_{m,m-2}^TX_{m,m-2}
    \end{bmatrix} \\
    &= \frac{1}{m!}\begin{bmatrix}
        m! & (m-1)!\textbf{1}^T \\
        (m-1)!\textbf{1} & (m-1)!I_q + (m-2)!V
    \end{bmatrix} \\
    &= \begin{bmatrix}
        1  & (1/m)\textbf{1}^T \\
        (1/m)\textbf{1}  & [(m-1)!I_q + (m-2)!V]/m!
    \end{bmatrix}
\end{align} which yields the desired result. This completes the proof.
    
\end{proof}

\begin{proof}[\textbf{Proof of Theorem \ref{thm:momentmat2}}]
    Let $D_m$ be the full design that uses all $m!$ permutations exactly once. Let $X_m$ be the model matrix expansion of $D_m$ under Model (\ref{eqn:oofamodel2}), with the constraints that $\beta_{m,m-1} = 0$ and $\delta_{m,m-1} = 0$. Let $q = 2{m \choose 2} - 1$. Partition $X_m = [X_0, X_m^{(1)}, X_m^{(2)}]$, where $X_0 = [1,\dots,1]^T$, $X_m^{(1)}$ is a $m! \times q$ matrix with columns corresponding to length one transition effects, and $X_m^{(2)}$ is a $m! \times q$ matrix with columns corresponding to length two transition effects. Notice that \begin{align}
\label{eqn:momentmat2}
    M = X_m^TX_m/m! = \frac{1}{m!}\begin{bmatrix}
        X_0^T X_0  & X_0^T X_m^{(1)}  & X_0^{T} X_m^{(2)} \\
        X_m^{(1) T} X_0  & \left(X_m^{(1)}\right)^T X_m^{(1)} & \left(X_m^{(1)}\right)^T X_m^{(2)} \\
        X_m^{(2) T} X_0  & \left(X_m^{(2)}\right)^T X_m^{(1)} & \left(X_m^{(2)}\right)^T X_m^{(2)}
    \end{bmatrix}
\end{align} 

From Theorem 1, we know that $X_0^TX_0 = m!$, $X_0^TX_m^{(1)} = (m-1)!\textbf{1}^T$, and $[X_m^{(1)}]^TX_m^{(1)} = (m-1)!I_q + (m-2)!V$. Therefore, we only need to determine the values of $X_0^{T}{} X_m^{(2)}$, $[X_m^{(1)}]^TX_m^{(2)}$, and $[X_m^{(2)}]^TX_m^{(2)}$ to complete the theorem. Let $X_{i,j}^{(1)}$ denote the column of $X_m^{(1)}$ that corresponds to the length one transition effect $\beta_{i,j}$. Let $X_{i,j}^{(2)}$ denote the column of $X_m^{(2)}$ that corresponds to the length two transition effect $\delta_{i,j}$. Let $i,j,k,\ell$ be four distinct elements of $\{1,\dots,m\}$. Then, the following are true: \begin{enumerate}
    \item[(a)] $X_0^TX_{i,j}^{(2)} = \sum_{\textbf{a} \in D_m} x_{i,j}^{(2)}(\textbf{a}) = (m-2)[(m-2)!]$. This is true because, for any permutation $\textbf{a}$, there are $m-2$ ways to place components $i$ and $j$ in $\textbf{a}$ that satisfy $d(i,j,\textbf{a}) = 2$; for each placement of components $i$ and $j$, there are $(m-2)!$ permutations of the remaining components. 
    \item[(b)] $[X_{i,j}^{(1)}]^TX_{i,j}^{(2)} = \sum_{\textbf{a} \in D_m} x_{i,j}(\textbf{a})x_{i,j}^{(2)}(\textbf{a}) = 0$ because for any $\textbf{a} \in D_m$, $d(i,j,\textbf{a})$ cannot be equal to both 1 and 2. Similarly, $[X_{i,j}^{(1)}]^TX_{j,i}^{(2)} = 0$. 
    \item[(c)] $[X_{i,j}^{(1)}]^TX_{i,\ell}^{(2)} = \sum_{\textbf{a} \in D_m} x_{i,j}(\textbf{a})x_{i,\ell}^{(2)}(\textbf{a}) = (m-2)[(m-3)!]$, because $x_{i,j}(\textbf{a}) = x_{i,\ell}^{(2)}(\textbf{a}) = 1$ implies that the sequence $(i,j,\ell)$ occurs in $\textbf{a}$. There are $(m-2)$ possible placements of this sequence. For each placement of $(i,j,\ell)$, there are $(m-3)!$ permutations of the remaining components. 
    \item[(d)] $[X_{i,j}^{(1)}]^TX_{k,j}^{(2)} = \sum_{\textbf{a} \in D_m} x_{i,j}(\textbf{a})x_{k,j}^{(2)}(\textbf{a}) = (m-2)[(m-3)!]$, because $x_{i,j}(\textbf{a}) = x_{k,j}^{(2)}(\textbf{a}) = 1$ implies that the sequence $(k,i,j)$ occurs in $\textbf{a}$. There are $(m-2)$ possible placements of this sequence. For each placement of $(k,i,j)$, there are $(m-3)!$ permutations of the remaining components. 
    \item[(e)] $[X_{i,j}^{(1)}]^TX_{j,\ell}^{(2)} = \sum_{\textbf{a} \in D_m} x_{i,j}(\textbf{a})x_{j,\ell}^{(2)}(\textbf{a}) = (m-3)[(m-3)!].$ Notice that $x_{i,j}(\textbf{a}) = x_{j,\ell}^{(2)}(\textbf{a}) = 1$ implies that the sequence $(i,j,b,\ell)$ is in $\textbf{a}$ for some $b \in \{1,\dots,m\} \setminus \{i,j,\ell\}$. For a given value of $b$, there are $m-3$ ways to place this sequence. For each $b$ and each placement of the sequence, there are $(m-4)!$ permutations of the remaining components. Since there are $m-3$ possible values of $b$, there are $(m-3)[(m-3)(m-4)!] = (m-3)[(m-3)!]$ permutations in $D_m$ that satisfy $x_{i,j}(\textbf{a}) = x_{j,\ell}^{(2)}(\textbf{a}) = 1$.
    \item[(f)] $[X_{i,j}^{(1)}]^TX_{k,i}^{(2)} = \sum_{\textbf{a} \in D_m} x_{i,j}(\textbf{a})x_{k,i}^{(2)}(\textbf{a}) = (m-3)[(m-3)!].$ Notice that $x_{i,j}(\textbf{a}) = x_{k,i}^{(2)}(\textbf{a}) = 1$ implies that the sequence $(k,b,i,j)$ is in $\textbf{a}$ for some $b \in \{1,\dots,m\} \setminus \{i,j,k\}$. For a given value of $b$, there are $m-3$ ways to place this sequence. For each $b$ and each placement of the sequence, there are $(m-4)!$ permutations of the remaining components. Since there are $m-3$ possible values of $b$, there are $(m-3)[(m-3)(m-4)!] = (m-3)[(m-3)!]$ permutations in $D_m$ that satisfy $x_{i,j}(\textbf{a}) = x_{k,j}^{(2)}(\textbf{a}) = 1$.
    \item[(g)] $[X_{i,j}^{(1)}]^TX_{k,\ell}^{(2)} = \sum_{\textbf{a} \in D_m} x_{i,j}(\textbf{a})x_{k,\ell}^{(2)}(\textbf{a}) = (m-4)[(m-3)!].$ Suppose component $i$ is placed directly before component $j$. This leaves $m - 2$ remaining positions to place component $k$. Two of these positions will not leave enough room to place component $\ell$ two positions to the right, so for a given placement of $i$ and $j$, there are only $m-4$ feasible positions for $k$ and $\ell$. There are $m-4$ remaining components to be placed; therefore, for a given position of $(i,j)$, there are $(m-4)!$ permutations in $D_m$ that satisfy $x_{k,\ell}^{(2)}(\textbf{a}) = 1$. Since $k$ and $\ell$ take up two positions, there are $m-3$ ways to place component $i$ and have room to place component $j$ directly after $i$. Therefore, there are a total of $(m-3)(m-4)[(m-4)!] = (m-4)[(m-3)!]$ ways to satisfy $ x_{i,j}(\textbf{a})= x_{k,\ell}^{(2)}(\textbf{a}) = 1$.
    \item[(h)] $[X_{i,j}^{(2)}]^TX_{i,j}^{(2)} = \sum_{\textbf{a} \in D_m} [x_{i,j}^{(2)}(\textbf{a})]^2 =   \sum_{\textbf{a} \in D_m} x_{i,j}^{(2)}(\textbf{a}) = (m-2)[(m-2)!]$ by part (a). 
    \item[(i)] $[X_{i,j}^{(2)}]^TX_{j,i}^{(2)} = \sum_{\textbf{a} \in D_m} x_{i,j}^{(2)}(\textbf{a})x_{j,i}^{(2)}(\textbf{a}) = 0$ because $i$ cannot be both before and after $j$.
    \item[(j)] $[X_{i,j}^{(2)}]^TX_{k,i}^{(2)} = \sum_{\textbf{a} \in D_m} x_{i,j}^{(2)}(\textbf{a})x_{k,i}^{(2)}(\textbf{a}) = (m-4)[(m-3)!].$ In this case, if $ x_{i,j}^{(2)}(\textbf{a}) = x_{j,\ell}^{(2)}(\textbf{a}) = 1$, then the sequence $(k,b,i,c,j)$ must be in $\textbf{a}$ for some components $b,c \in \{1,\dots,m\} \setminus \{i,j,k\}$. For fixed values of $b$ and $c$, there are $m-4$ ways to place this sequence. For each of these placements, there are $(m-5)!$ permutations of the remaining components. There are $m-3$ possible values for $b$ (since $b$ can be in $1,\dots,m$, but $b \neq i,j,k$), and given $b$, there are $m-4$ possible values for $c$. Therefore, $\sum_{\textbf{a} \in D_m} x_{i,j}^{(2)}(\textbf{a})x_{k,i}^{(2)}(\textbf{a}) = (m-3)(m-4)(m-4)[(m-5)!] = (m-4)[(m-3)!]$.
    \item[(k)] $[X_{i,j}^{(2)}]^TX_{j,\ell}^{(2)} = \sum_{\textbf{a} \in D_m} x_{i,j}^{(2)}(\textbf{a})x_{j,\ell}^{(2)}(\textbf{a}) = (m-4)[(m-3)!].$  In this case, if $x_{i,j}^{(2)}(\textbf{a}) = x_{j,\ell}^{(2)}(\textbf{a}) = 1$, then the sequence $(i,b,j,c,\ell)$ must occur in $\textbf{a}$ for some $b,c \in \{1,\dots,m\} \setminus \{ i,j,\ell \}$. By the same counting argument in part (j), the result follows.
    \item[($\ell$)]  $[X_{i,j}^{(2)}]^TX_{k,\ell}^{(2)} = \sum_{\textbf{a} \in D_m} x_{i,j}^{(2)}(\textbf{a})x_{k,\ell}^{(2)}(\textbf{a}) = [(m-6)(m-5) + 4(m-4)][(m-4)!]$. This is true because there are $m-2$ possible positions to place component $i$ so that component $j$ may be placed two spaces to the right of $i$. Of these $m-2$ possibilities, consider the following four placements of $i$ and $j$, \begin{align*}
        &i, a_1, j, a_2, \dots, a_{m-2} \\
        &a_1, i, a_2, j, \dots, a_{m-2} \\
        &a_1, a_2, \dots, a_{m-4}, j, a_{m-3}, i, a_{m-2} \\
        &a_1, a_2, \dots, a_{m-3}, j, a_{m-2}, i, 
    \end{align*} where $a_1, \dots, a_{m-2}$ represent available positions. In these four sequences, if component $k$ is placed in one of the first $m-2-2 = m - 4$ available positions, then there is space to place component $\ell$ two spaces to the right of it. Once the 4 components $i, j, k, \ell$ have been placed, there are $(m-4)!$ permutations of the remaining $m-4$ components. 

    In the remaining $m - 2 - 4 = m - 6$ possible positions of $i$ and $j$, there are only $m - 5$ ways to place $k$ so that $j$ is exactly two spaces to the right of it (as $i$ or $j$ will take up one of the remaining placements). Once the 4 components $i, j, k, \ell$ have been placed, there are $(m-4)!$ permutations of the remaining $m-4$ components.  

    Overall, these statements imply that \begin{align}
        [X_{i,j}^{(2)}]^TX_{k,\ell}^{(2)} &= 4(m-4)[(m-4)!] + (m-6)(m-5)[(m-4)!] \\ &= [(m-6)(m-5) + 4(m-4)][(m-4)!].
    \end{align}
\end{enumerate} The equation in (a) shows that $X_0^{T}{} X_m^{(2)} = (m-2)[(m-2)!]\textbf{1}_q$. The equations in (b) to (g) show that $[X_m^{(1)}]^TX_m^{(2)} = [(m-3)!]Q$. The remaining equations show that $[X_m^{(2)}]^TX_m^{(2)} = (m-2)(m-2)!I_q + R$. It follows from (\ref{eqn:momentmat2}) that \begin{align}
    M = \begin{bmatrix}
        1  & (1/m)\textbf{1}^T & \frac{m-2}{m(m-1)}\textbf{1}^T  \\
        (1/m)\textbf{1}  & [(m-1)!I_q + (m-2)!V]/m! & [(m-3)!/m!]Q \\
        \frac{m-2}{m(m-1)}\textbf{1} & [(m-3)!/m!]Q^T   & [(m-2)(m-2)!I_q + R]/m!
    \end{bmatrix},
\end{align} which completes the proof. 
\end{proof}

\begin{proof}[\textbf{Proof of Corollary \ref{cor:momentmat3}}]
    Suppose that the $m$ components are arranged into $c$ blocks $b_1, \dots, b_c$ such that if $i < j$, then all components in $b_i$ must precede all components in $b_j$. As stated in Corollary \ref{cor:momentmat3}, label the components so that components $1,\dots,n_1$ belong to $b_1$, components $n_{i-1} + 1, \dots, n_i$ belong to block $b_i$ for $i = 2,\dots,c-1$, and $n_{c-1} + 1, \dots, m$ belong to the final block $b_c$. Let $D_m$ be the full design that uses each of the $N = \prod_{i=1}^c (m_i)!$ feasible permutations exactly once, where $m_i$ is the number of components in block $b_i$, $i = 1,\dots,c$. Let $q_i = 2{m_i \choose 2} - 1$, and let $V_i$ be a $q_i \times q_i$ matrix that is defined as in Theorem \ref{thm:momentmat1}, and let $I_{q_i}$ denote the identity matrix of dimension $q_i$. Let $\textbf{1}$ be a conformable vector of ones, and $\textbf{0}$ be a $q_i \times q_i$ matrix of zeroes. Let $\textbf{1}_{q_i \times q_j}$ be a $q_i \times q_j$ matrix of ones.

    Partition $X_m = [\mathbf{1}, X_1, X_2, \dots, X_c]$ where $X_i$ is an $N \times q_i$ matrix whose columns correspond to the $q_i$ length one transitions between components that are exclusively in block $b_i$, $i = 1,\dots,c$. It follows that the moment matrix is \begin{align}
        \label{eqn:cormomentmat1}
        M = X_m^TX_m/N = \frac{1}{N}\begin{bmatrix}
        \textbf{1}^T\textbf{1} & \textbf{1}^TX_1 & \textbf{1}^TX_2 & \dots &\textbf{1}^TX_c \\
        X_1^T\textbf{1} & X_1^TX_1 & X_1^TX_2 & \dots & X_1^TX_c \\
        X_2^T\textbf{1} & X_2^TX_1 & X_2^TX_2 & \dots & X_2^TX_c \\
        \vdots & \vdots & \vdots & \ddots & \vdots \\
        X_c^T\textbf{1} & X_c^TX_1 & X_c^TX_2 & \dots & X_c^TX_c 
    \end{bmatrix}.
    \end{align} 
    By Theorem \ref{thm:momentmat1}, we have that $X_i^TX_i = [(m_i -1)!I_{q_i} + (m_i - 2)!V_i]$ for $i = 1,\dots,c$. Notice that for components $j$ and $k$ belonging to the same block $b_i$, we have that \begin{align*}
        \frac{1}{N}\sum_{\textbf{a} \in D_m} x_{j,k}(\textbf{a}) &= \frac{(m_i - 1)!\prod_{\ell = 1,\ell \neq i}^c(m_{\ell}!)}{N} = \frac{(m_i - 1)!\prod_{\ell = 1,\ell \neq i}^c(m_{\ell}!)}{\prod_{\ell = 1}^c m_{\ell}!} \\
        &= \frac{(m_i-1)!}{m_i!} = \frac{1}{m_i},
    \end{align*} since $j$ appears directly in front of $k$ exactly $(m_i - 1)!$ times by item (b) of Theorem \ref{thm:momentmat1}, and there are $\prod_{\ell = 1,\ell \neq i}^c(m_{\ell}!)$ possible permutations of the other $c-1$ blocks. It follows that $\textbf{1}^TX_i/N = (1/m_i)\textbf{1}^T$ for $i = 1,\dots,c$. Similarly, it follows that for $i \neq j$, $X_i^TX_j/N = (m_i-1)!(m_j-1)!\textbf{1}_{q_i \times q_j}/N = \frac{1}{m_im_j}\textbf{1}_{q_i \times q_j}$. Finally, notice that $\textbf{1}^T\textbf{1} = N$.  Substituting these values into the matrix in (\ref{eqn:cormomentmat1}) yields \begin{align}
         X_m^TX_m/N = \begin{bmatrix}
        1  & (1/m_1)\textbf{1}^T &  (1/m_2)\textbf{1}^T  & \dots & (1/m_c)\textbf{1}^T  \\
        (1/m_1)\textbf{1}  & M_1 & \frac{1}{m_1m_2}\textbf{1}_{q_1 \times q_2} & \dots & \frac{1}{m_1m_c}\textbf{1}_{q_1 \times q_c} \\
        (1/m_2)\textbf{1}  & \frac{1}{m_2m_1}\textbf{1}_{q_2 \times q_1}  & M_2 & \dots, & \frac{1}{m_2m_c}\textbf{1}_{q_2 \times q_c}\\
        \vdots & \vdots & \vdots & \ddots & \vdots \\
        (1/m_c)\textbf{1}  & \frac{1}{m_cm_1}\textbf{1}_{q_c \times q_1} & \frac{1}{m_cm_2}\textbf{1}_{q_c \times q_2} & \dots & M_c
    \end{bmatrix},
    \end{align} where $M_i = [(m_i -1)!I_{q_i} + (m_i - 2)!V_i]/N$ for $i = 1,\dots,c$. This completes the proof.
    
\end{proof}

\end{document}